\documentclass[letterpaper]{article} 
\usepackage{aaai25}  
\usepackage{times}  
\usepackage{helvet}  
\usepackage{courier}  
\usepackage[hyphens]{url}  
\usepackage{graphicx} 
\usepackage{natbib}  
\usepackage{caption} 
\usepackage{amsmath} 
\usepackage{amsfonts} 
\usepackage{mathtools} 
\usepackage{amsthm} 
\usepackage{dsfont} 
\usepackage{empheq} 
\usepackage{multirow} 
\usepackage{arydshln} 
\usepackage{amssymb}
\usepackage{tikz} 
\usetikzlibrary{positioning}

\usepackage[capitalize,noabbrev]{cleveref} 
\crefname{relctr}{relation}{Relation} 

\newtheorem{proposition}{Proposition}
\newtheorem{lemma}{Lemma}
\newtheorem{definition}{Definition}

\newtheorem{remark}{Remark}
\newtheorem{fact}{Basic fact}
\newtheorem{assumption}{Assumption}
\newtheorem{theorem}{Theorem}
\newtheorem{corollary}[theorem]{Corollary}
\newtheorem{example}{Example}
\newtheorem{condition}[assumption]{Condition~C\ignorespaces}

\providecommand{\R}{\mathbb{R}} 
\providecommand{\N}{\mathbb{N}} 

\providecommand{\qq}{\mathbf{q}}
\providecommand{\0}{\mathbf{0}}
\providecommand{\1}{\mathbf{1}}
\providecommand{\bb}{\mathbf{b}}
\providecommand{\ee}{\mathbf{e}}
\providecommand{\vv}{\mathbf{v}}
\providecommand{\xx}{\mathbf{x}}

\providecommand{\mE}{\mathbf{E}}
\providecommand{\mI}{\mathbf{I}}
\providecommand{\mM}{\mathbf{M}}
\providecommand{\mQ}{\mathbf{Q}}

\providecommand{\cB}{\mathcal{B}}
\providecommand{\cI}{\mathcal{I}}
\providecommand{\cJ}{\mathcal{N}}
\providecommand{\cM}{\mathcal{M}}
\providecommand{\cQ}{\mathcal{Q}}

\providecommand{\T}{^{\top}} 

\renewcommand{\tilde}{\widetilde} 

\providecommand\ie{\emph{i.e.}}
\providecommand\eg{\emph{e.g.}}
\providecommand\st{\textit{s.t. }}

\title{Linear Equations with Min and Max Operators: Computational Complexity}

\author {
Krishnendu Chatterjee\textsuperscript{\rm 1},
Ruichen Luo\textsuperscript{\rm 1},
Raimundo Saona\textsuperscript{\rm 1},
Jakub Svoboda\textsuperscript{\rm 1}
}
\affiliations {
\textsuperscript{\rm 1}Institute of Science and Technology Austria\\
Am Campus 1\\
Klosterneuburg, 3400 Austria\\
krishnendu.chatterjee@ist.ac.at,
rluo@ist.ac.at,
raimundo.saona@gmail.com,
jakub.svoboda@ist.ac.at
}

\begin{document}

\maketitle

\begin{abstract}
	We consider a class of optimization problems defined by a system of linear equations with min and max operators.
	This class of optimization problems has been studied under restrictive conditions, such as,
	(C1)~the halting or stability condition;
	(C2)~the non-negative coefficients condition;
	(C3)~the sum up to~$1$ condition; and
	(C4)~the only min or only max operator condition.
	Several seminal results in the literature focus on special cases.
	For example, turn-based stochastic games correspond to conditions C2 and C3; and Markov decision process to conditions C2, C3, and C4.
	However, the systematic computational complexity study of all the cases has not been explored, which we address in this work.
	Some highlights of our results are: with conditions C2 and C4, and with conditions C3 and C4, the problem is \textsc{NP}-complete, whereas with condition C1 only, the problem is in \textsc{UP} $\cap$ \textsc{coUP}.
	Finally, we establish the computational complexity of the decision problem of checking the respective conditions.
\end{abstract}

\section{Introduction}

\noindent\textbf{Optimization problems.}
Optimization problems are in the core of most applications of artificial intelligence.
Prominent examples include:
learning algorithms, where stochastic gradient descent approaches are fundamental~\cite{bottou2010large};
game theory, where the connection between matrix games and linear programming was established in the seminal work of~\citeauthor{von1953theory};
sequential decision making, where the solution is obtained via Markov decision processes which can be solved via linear programming~\cite{puterman2014markov}.
Hence, understanding the computational complexity of different classes of optimization problems is a fundamental question in artificial intelligence and learning and an active research area, for example, linear programming~\cite{dantzig2002linear}, convex optimization~\cite{nesterov2013introductory}, semi-definite programming~\cite{gartner2012approximation}, to name a few.

\noindent\textbf{Linear equations with min-max operator.}
In this work, we consider the class of optimization problems defined by a system of linear equations along with min and max operators.
Since linear equalities are very basic operations, and min and max operators are the most fundamental optimization operators, this corresponds to a very natural class of optimization problems, for example, this is a natural generalization of linear equations with Boolean variables.

\noindent\textbf{Restrictive conditions.}
Given the generality of the above class, the associated decision problems are computationally hard.
Hence, several natural sub-classes of the above optimization problem have been studied in the literature and we first recall the restrictive conditions:
(C1)~the halting or the stability condition;
(C2)~the non-negative coefficients condition;
(C3)~the sum up to~$1$ condition; and
(C4)~the min-only or max-only operator condition.

\noindent\textbf{Classical complexity results from the literature.}
Several seminal results consider subsets of the above conditions.
First, conditions C2 and C3 together represent stochasticity or probability distributions.
With these two conditions and both max and min operators, we obtain the very well-studied class of simple stochastic games (SSGs).
SSGs have been studied in the seminal work of~\citet{condon1992complexity} as an important sub-class of Shapley games~\cite{shapley1953stochastic}.
The computational problem of SSGs is known to be in \textsc{NP} $\cap$ \textsc{coNP}~\cite{condon1992complexity}, as well as \textsc{UP} $\cap$ \textsc{coUP}\footnote{A language is in \textsc{UP} (or \textsc{coUP}) if a given \textsc{Yes} (or \textsc{No}) answer can be verified in polynomial time, and the verifier machine only accepts at most one answer for each problem instance~\cite{hemaspaandra1997unambiguous}.}~\citep{chatterjee2011reduction}. The existence of polynomial-time algorithms for SSGs is a major and long-standing open problem in game theory. We mention a problem is SSG-hard if there is a polynomial-time reduction from the SSGs; hence for SSG-hard problems polynomial-time algorithms require a major breakthrough. If we have conditions C2, C3, and C4, then this corresponds to the class of Markov decision processes (MDPs), where there is a single player, as opposed to SSGs where there are two adversarial players.
Hence, under conditions C2, C3, and C4 the problems are polynomial-time solvable via linear programming.

\begin{table*}[ht]
	\centering
	\begin{tabular}{ |l|c|c| }
		\hline
		${\textbf{\{} \textbf{C\ref{condition:non-negative},C\ref{condition:one-type}} \textbf{\}}}$,
		${ \textbf{\{} \textbf{C\ref{condition:sum-to-1},C\ref{condition:one-type}} \textbf{\}}}$,
		${\{ \textnormal{C\ref{condition:non-negative}} \}}$,
		${\{ \textnormal{C\ref{condition:sum-to-1}} \}}$, ${\{ \textnormal{C\ref{condition:one-type}} \}}$,
		${\emptyset}$
		& \textsc{NP-complete} & \cref{thm:np complete results} \\
		\hline
		${ \textbf{\{} \textbf{C\ref{condition:conv_halting},C\ref{condition:sum-to-1},C\ref{condition:one-type}} \textbf{\}} }$,
		${ \textbf{\{} \textbf{C\ref{condition:conv_halting},C\ref{condition:sum-to-1}} \textbf{\}} }$,
		${ \textbf{\{} \textbf{C\ref{condition:conv_halting},C\ref{condition:one-type}} \textbf{\}} }$,
		${ \textbf{\{} \textbf{C\ref{condition:conv_halting}} \textbf{\}} }$
		& \multirow{3}{*} { \begin{tabular}[]{@{}c@{}} \textsc{UP} $\cap$ \textsc{coUP} \\ (SSG-hard) \end{tabular} } & \multirow{2}{*} { \begin{tabular}[]{@{}c@{}} Corollary~\ref{cor:up cap coup results}~\&~\ref{cor:equivalence and ssg hardness} \end{tabular} }  \\
		\cdashline{1-1}
		${ \textbf{\{} \textbf{C\ref{condition:conv_halting},C\ref{condition:non-negative}} \textbf{\}} }$
		& & \\
		\cdashline{1-1} \cdashline{3-3}
		$\{ \textnormal{C\ref{condition:conv_halting},C\ref{condition:non-negative},C\ref{condition:sum-to-1}} \}$,
		$\{ \textnormal{C\ref{condition:non-negative},C\ref{condition:sum-to-1}} \}$
		& & \cref{thm:ssg complexity} \\
		\hline
		$\{ \textnormal{C\ref{condition:conv_halting},C\ref{condition:non-negative},C\ref{condition:sum-to-1},C\ref{condition:one-type}} \}$,
		$\{ \textnormal{C\ref{condition:non-negative},C\ref{condition:sum-to-1},C\ref{condition:one-type}} \}$,
		${ \textbf{\{} \textbf{C\ref{condition:conv_halting},C\ref{condition:non-negative},C\ref{condition:one-type}} \textbf{\}} }$
		& \textsc{PTIME} & \cref{cor:PTIME results} \\
		\hline
	\end{tabular}
	\caption{%
		The complexity of the LEMM decision problems under all subsets of conditions.
		{If $X \subseteq Y$ represents two subsets of conditions, then the LEMM decision problem under $X$ that has fewer conditions is more general.} 
		The table describes all the results, with our main results in bold. Moreover, for the problems in a row there is a polynomial-time equivalence.
	}
	\label{tab:complexity results}
\end{table*}

\begin{table*}[ht]
	\centering
	\begin{tabular}{|l|c|c|}
		\hline
		{$ \textbf{\{} \textbf{C\ref{condition:conv_halting}}  \textbf{\}} $},
		{$ \textbf{\{} \textbf{C\ref{condition:conv_halting},C\ref{condition:sum-to-1}}  \textbf{\}} $},
		{$ \textbf{\{} \textbf{C\ref{condition:conv_halting},C\ref{condition:one-type}}  \textbf{\}} $},
		{$ \textbf{\{} \textbf{C\ref{condition:conv_halting},C\ref{condition:sum-to-1},C\ref{condition:one-type}}  \textbf{\}} $} & \textsc{coNP}-hard & \cref{cor:coNP-hard checks} \\
		\hline
		{$ \textbf{\{} \textbf{C\ref{condition:conv_halting},C\ref{condition:non-negative}}  \textbf{\}} $},
		{$ \textbf{\{} \textbf{C\ref{condition:conv_halting},C\ref{condition:non-negative},C\ref{condition:one-type}}  \textbf{\}} $},
		$\{\textnormal{C\ref{condition:conv_halting},C\ref{condition:non-negative},C\ref{condition:sum-to-1}}\}$,
		$\{\textnormal{C\ref{condition:conv_halting},C\ref{condition:non-negative},C\ref{condition:sum-to-1},C\ref{condition:one-type}}\}$ & \textsc{PTIME} & \cref{cor:PTIME checks} \\
		\hline
	\end{tabular}
	
	\caption{The complexity of the condition decision problems for all subsets of conditions in the presence of condition~C\ref{condition:conv_halting}.
		In the absence of condition C1, all other conditions can be checked in linear time.
		Our results are in bold.}
	\label{tab:condition decision}
\end{table*}

\noindent\textbf{Classical results from the literature related to various connections.}
It was shown in~\cite{condon1992complexity} that, if we have conditions C2 and C3, then these two conditions together imply condition C1 as well.
Moreover, recent results show that robust versions of MDP problems can be reduced to SSGs~\cite{chatterjee2024solving}.
Thus, robust versions of optimization problems with only max or only min operators naturally give rise to both operators.

\noindent\textbf{Open problems.}
While there are several studies that consider special sub-classes of the above problem, a systematic study of the computational complexities for all cases has been missing in the literature and is the focus of this work.
For example, condition C1 is a natural stability condition, and the complexity of this problem without the restriction of C2 and/or C3 has not been addressed in the literature.

\noindent\textbf{Our contributions.}
In this work, we present the computational complexity landscape for the class of optimization problem of linear equations with min and max operators under all subsets of restrictive conditions.
In particular, some highlights of our results are as follows:
(i)~with conditions C2 and C4, and with conditions C3 and C4, the problem is \textsc{NP}-complete;
(ii)~with conditions C1 only, the problem is in \textsc{UP} $\cap$ \textsc{coUP};
(iii)~even with conditions C1, C3, and C4, the problem is SSGs-hard and hence proving the existence of a polynomial-time algorithm requires a major breakthrough.
These complexity results are summarized in \Cref{tab:complexity results}.
Finally, we consider the complexity of the decision problem of checking the conditions.
We show that checking condition C1 is \textsc{coNP}-hard, but checking conditions C1 and C2 can be achieved in polynomial time. These complexity results are summarized in \Cref{tab:condition decision}.

\noindent\textbf{Technical contributions.}
Our main technical contributions are as follows.
First, we show that if we have only condition C1, but not condition C2 and C3, then several classical properties of SSGs, e.g., monotonicity of solutions and the $\min \max = \max \min$ property do not hold.
We provide illuminating examples (see Example~\ref{example:non-monotonicity}--Example~\ref{example:misleading subsolution}) to illustrate this aspect.
Second, even though these fundamental properties break, we still show the existence of a unique solution under condition C\ref{condition:conv_halting} (see Lemma~\ref{lemma:existence&uniqueness}), which allows us to establish the \textsc{UP} $\cap$ \textsc{coUP} result.
Details and proofs omitted due to lack of space are presented in the Appendix.

\section{Definitions}

In this section, we present the basic definitions.
We define the linear equations with min and max operators (LEMM), the associated decision problem, and finally the conditions. We use the following notation: for an integer $k$, we denote $[1, k] \cap \mathbb{N}$ by $[k]$.

\begin{definition}[Linear equations with min and max operators (LEMM)]
	\label{def:lemm}
	Consider $(n_1,n_2,n,\cJ,\qq,\bb)$ such that the following conditions hold:
	(a)~$n_1, n_2, n \in \mathbb{N}$,
	(b)~$n \ge n_1+n_2$,
	(c)~$\emptyset \subsetneq \cJ (i) \subseteq [n]$ for $i \in [n_1+n_2]$, (d)~$\qq_k \in \R^n$ for $n_1+n_2 < k \le n$, and (e)~$\bb = [0, \ldots, 0, b_{n_1+n_2+1}, \ldots, b_{n}] \T \in \R^n$. A vector  $\xx = [x_1, \ldots, x_n] \T \in \R^n$ that satisfies the following system of linear equations with min and max operators (LEMM) is called feasible:

	\begin{subequations} \label{eq:Problem}
		\begin{empheq}[left = \empheqlbrace]{alignat=2}
			x_i &= \min _{l \in \cJ (i)} \ x_{l} , &\quad & 1 \le i \le n_1, \label{eq:min line} \\
			x_j &= \max _{l \in \cJ (j)} \ x_{l} , &\quad & n_1 < j \le n_1+n_2, \label{eq:max line} \\
			x_k &= \qq_k\T \xx + b_k, &\quad & n_1 + n_2  < k \le n . \label{eq:affine line}
		\end{empheq}
	\end{subequations}
\end{definition}

In \cref{eq:Problem}, $x_1, \ldots, x_{n_1}$ in \eqref{eq:min line} are the min variables, $x_{n_1+1}, \ldots, x_{n_1+n_2}$ in \eqref{eq:max line} are the max variables, and $x_{n_1+n_2+1}, \ldots, x_{n}$ in \eqref{eq:affine line} are the affine variables.

\begin{definition}[LEMM decision problem]
	The LEMM decision problem is: given an LEMM with $(n_1,n_2,n,\cJ,\qq,\bb)$, an index $i \in [n]$, and a threshold $\beta \in \R$, determine whether there exists a feasible solution $\xx$ to the LEMM, \st $x_i < \beta$.
\end{definition}

\begin{remark}[Generalities]
	We discuss the generality of the LEMMs.
	\begin{itemize}
		\item Any finite nested min, max, and linear operators can be separated by substitution. For instance, consider the following equation
		$$
		z_i = \min _{j \in \cJ} \ f_{i,j}(z) ,
		$$
		where $f_{i,j}$ contains finite (nested) min, max, and linear operators.
		It can be rewritten as
		$$
		z_i = \min _{j \in \cJ} \ x_{i,j} \,, \qquad x_{i,j} = f_{i,j} (z) \,.
		$$
		
		\item Boolean variables can be encoded using the min and/or max, and linear operators (see \eg\ \cref{eq:reduction to partition problem} in the Appendix), and hence LEMMs generalize
		linear equations with Boolean variables.
		
	\end{itemize}
\end{remark}

While the LEMM decision problem is quite general, which implies computational hardness, several restrictive sub-classes have been considered in the literature.
We first introduce some notations and then describe the conditions.

\noindent\textbf{Notations.}
We introduce the following notations.
Let the indicator vector
$$
\ee_i \coloneqq [ \delta_{i,1}, \ldots, \delta_{i,n} ]\T ,\quad i \in [n] \,,
$$
where \(\delta_{i,j}\) is $1$ if $i = j$ and $0$ otherwise.
Denote the set system of linear equations $\xx = \mQ \xx + \bb$ induced by forcing each min and max variable $x_i$ to be equal to $x_{\ell(i)}$, i.e., fixing one choice for each min and max variable, by
\begin{align*}
	\cQ
	&\coloneqq
	\{ \mQ_i \}_{i \in \mathcal{I}} \\
	&\coloneqq \Big\{ [\ee_{\ell(1)}, \ldots, \ee_{\ell(n_1+n_2)}, \qq_{n_1+n_2+1}, \ldots, \qq_{n}] \T \\
	&\qquad\qquad \Big|\; \ell (j) \in \cJ (j) \text{ for all } j \in [n_1+n_2] \Big\} \,.
\end{align*}
The convex hull of $\cQ$ is denoted as $\mathbf {conv} (\cQ)$:
$$
\left\{ \sum _{i \in \cI} \alpha_i \mQ_i \;\middle|\;
\sum _{i \in \cI} \alpha_i = 1
\text{ and }
\forall\ i \in \cI,\ \mQ_i \in \cQ \land \alpha_i \ge 0  \right\} .
$$

\begin{condition}[Halting or stability]
	\label{condition:conv_halting}
	For all $\mQ \in \mathbf {conv} (\cQ)$,
	$$
	\lim _{m \rightarrow \infty} \mQ ^m = \0_{n \times n}.
	$$
\end{condition}

\begin{condition}[Non-negative coefficients]
	\label{condition:non-negative}
	For all $n_1 + n_2  < k \le n$, we have that $\qq_k \ge 0$ and $b_k \ge 0$.
\end{condition}

\begin{condition}[Sum up to $1$]
	\label{condition:sum-to-1}
	For all $n_1 + n_2  < k \le n$, we have that $\qq_k \T \1 + b_k \le 1$.
\end{condition}

\begin{condition}[Max-only or min-only]
	\label{condition:one-type}
	Either $n_1=0$ or $n_2=0$.
\end{condition}

We will use the following subset notation for the LEMM decision problem under various subsets of the conditions, \eg, ``the LEMM decision problem under $\{\textnormal{C\ref{condition:conv_halting},C\ref{condition:non-negative}}\}$'' means ``the LEMM decision problem under conditions C\ref{condition:conv_halting} and C\ref{condition:non-negative}''.

\begin{remark}[Relevance of the conditions in previous studies]
	\label{remark: Relevance of the conditions in previous studies}
	We clarify the relevance of the conditions. 
	\begin{enumerate}
		\item 
		The classical turn-based version of stochastic games played by two adversarial players with reachability objectives (which is referred to as simple stochastic games or SSGs) has stochastic or probabilistic transitions~\cite{condon1990algorithms}.
		The stochastic transitions lead to conditions~C\ref{condition:non-negative}~and~C\ref{condition:sum-to-1}, and stochastic games require both min and max operators for the two players.
		Hence, SSGs naturally correspond to LEMM with conditions~C\ref{condition:non-negative}~and~C\ref{condition:sum-to-1}.
		
		\item 
		In the context of SSGs (\ie, if conditions~C\ref{condition:non-negative}~and~C\ref{condition:sum-to-1}  hold), then we also obtain condition~C\ref{condition:conv_halting} without loss of generality~\cite{condon1992complexity}.
		
		\item 
		For SSGs, the existence of a polynomial-time algorithm is a major open problem, and we say a problem is SSG-hard if there is a polynomial-time reduction from SSGs. For SSG-hard problems, polynomial-time algorithms require a major breakthrough.
		
		\item 
		The intuitive description of condition C1 is as follows: in the absence of min and max operators, it is similar to Markov chains, and the condition implies that eventually recurrent states are reached almost-surely. This represents reaching the stable distribution almost-surely.
		In the presence of mix and max choices, this represents that irrespective of the choices, almost-surely stability is achieved. This is also referred to as the halting or stopping condition in the literature.
	\end{enumerate}
\end{remark}

We consider the problem of checking the conditions.

\begin{definition}[Condition decision problem]
	The condition decision problem is: given an LEMM and a subset of conditions, determine whether all the conditions are satisfied.
\end{definition}

\begin{remark}\label{remark:condition}
	For the condition decision problem, the conditions \textnormal{C\ref{condition:non-negative}, C\ref{condition:sum-to-1}, and C\ref{condition:one-type}} are easy to check. The main condition decision problem is in the presence of condition \textnormal{C\ref{condition:conv_halting}}.
\end{remark}

\section{Motivating examples}
\label{Section: Motivating examples}

In this section we present motivating
examples that can be modeled in the LEMM framework.

\noindent\textbf{Neural network verification.}
Consider the following \emph{neural network verification} problem.
Multilayer perceptrons (MLP) receive input and forward it by multiple layers of fully connected neurons with activation functions (e.g., ReLU or Maxout)~\cite{goodfellow2013maxout}.
The last layer is interpreted as a decision of the MLP.
For a bounded region of input, we want to decide whether the decision of the MLP is in a bounded region.

Formally, for input $\xx^{\textnormal{in}} \in [0,1]^d$, we are to decide whether the output satisfies  $x^{\textnormal{out}} < \beta$. The above problem can be posed as the following LEMM: $$
\left\{
\begin{aligned}
	\xx^{\textnormal{in}} &= \max \{ 0, \min \{ \xx^{\textnormal{in}}, 1 \} \} \,, \\
	\xx^{\textnormal{Layer }1} &= f ( \mQ^{\textnormal{Layer }1} \xx^{\textnormal{in}} + \bb^{\textnormal{Layer }1} ) \,, \\
	& \vdots \\
	\xx^{\textnormal{Layer }n} &= f ( \mQ^{\textnormal{Layer }n} \xx^{\textnormal{Layer }n-1} + \bb^{\textnormal{Layer }n} )\,, \\
	x^{\textnormal{out}} &= {\qq ^{\textnormal{out}}} \T \ \xx ^{\textnormal{Layer }n} + b^{\textnormal{out}} \,,
\end{aligned}
\right.
$$
where $f$ can be any piecewise linear activation function (\eg, ReLU or Maxout).
This problem is known to be \textsc{NP}-complete \cite{katz2017reluplex}, and corresponds to a LEMM decision problem with no restriction.
\qed

\noindent\textbf{Capital preservation.}
Motivated by the classical portfolio management problem~\cite{markowitz1952PortfolioSelection,puterman2014markov}, the capital preservation problem models the periodic management of assets under uncertainty.
Consider a cyclic evolution with $T$ periods, each period corresponds to an opportunity for the controller to affect the evolution of the assets, and for the market to affect the assets.
In each period, these effects come in turn, so at period $t \in [T]$ the controller can choose between different transformations that convert the assets from the previous stage in expectation in a linear fashion; i.e.,
having $x$ assets, we can get $y$ assets given by
$$
y \gets \max \{ q_{1, i}\, x + b_{1, i} \mid i \in [m] \} \,.
$$
Then, the uncertainty of the market is modeled by one of many possible linear transformations;
i.e., the controller obtains at the end of period $t$ an amount $z$ given by
$$
z \gets \min \{ q_{2, j}\, y + b_{2, j} \mid j \in [\ell] \} \,.
$$
Note that the effect of inflation on the value of the assets can be incorporated in, for example, $q_{2, j}$.
Following this dynamic, the controller chooses how to evolve the amount of assets at each period, and the cycle restarts.

Formally, this problem corresponds to the LEMM
$$
\left\{
\begin{aligned}
	x_{1, 1} &= \max \{ q^1_{1, i}\, x_{T, 2} + b^1_{1, i} \mid i \in [m] \} \,, \\
	x_{1, 2} &= \min \{ q^1_{2, j}\, x_{1, 1} + b^1_{2, j} \mid j \in [\ell] \} \,, \\
	\vdots &  \\
	x_{T, 1} &= \max \{ q^T_{1, i}\, x_{T-1, 2} + b^T_{1, i} \mid i \in [m] \} \,, \\
	x_{T, 2} &= \min \{ q^T_{2, j}\, x_{T, 1} + b^T_{2, j} \mid j \in [\ell] \} \,.
\end{aligned}
\right.
$$

We argue that we have conditions~C\ref{condition:conv_halting}~and~C\ref{condition:non-negative} in the following sense: condition~C\ref{condition:conv_halting} ensures that the controller does not get infinite returns in expectation; and condition~C\ref{condition:non-negative} implies that the assets are non-negative.
\qed

\noindent\textbf{Co-evolution in ecosystem.}
Branching processes~\cite{athreya2004branching} describe the evolution of population where the next state depends only on the previous state.
We consider the following \emph{co-evolution} extension of the problem.
There are $k$ different species in an ecosystem.
Evolution is the combined consequence of internal interactions (\eg, reproduction or predation) and external interventions (\eg, immigration or fertility control).
The population of the next generation is modeled as a piece-wise linear function of the current population.
We want to find the stable distribution of the ecosystem.

Formally, let $\xx \in \R_{\ge 0}^k$ be the current population, and the population of the next generation is given by
$$
\xx^\prime = \max\ \{ \mQ \xx + \bb , 0 \} \,,
$$
in which $\mQ \in \R ^{k \times k}$ models the internal interactions and $\bb \in \R^k$ models the external interventions. Hence the stable distribution of the ecosystem can be posed as the LEMM
$$
\xx = \max\ \{ \mQ \xx + \bb , 0 \} \,.
$$

We argue that we naturally have conditions~C\ref{condition:conv_halting}~and~C\ref{condition:one-type}: condition~C\ref{condition:conv_halting} implies that all species will go extinct without external intervention; and condition~C\ref{condition:one-type} holds because the LEMM is max-only.
\qed

The above examples show that LEMMs with restrictions can model classical optimization problems from the literature and motivate the computational complexity landscapes of various subsets of conditions.

\section{Complexity of LEMM decision problems}

We start with basic results and previous results from the literature, then we present our complexity bounds, and finally the equivalence between some problems.
In what follows, we consider that inputs of the problems are given as rational numbers as usual.

\subsection{Basic and previous results}
\label{subsec:basic results}

\begin{fact}
	\label{basic fact}
	If $X \subseteq Y$ represents two subsets of conditions, then the LEMM decision problem under $X$ is no easier than the LEMM decision problem under $Y$. Hence any complexity lower bound for $Y$ also holds for $X$, and any complexity upper bound for $X$ also holds for $Y$.
	See Figure~\ref{Figure: lattice} for the lattice structure of the LEMM decision problem under different conditions.
\end{fact}

\begin{figure}[h]
	\centering
	\begin{tikzpicture}[]
		\node at (3, -3) (1234) { \{C\ref{condition:conv_halting},C\ref{condition:non-negative},C\ref{condition:sum-to-1},C\ref{condition:one-type}\} };
		\node at (0, -2) (123) { \{C\ref{condition:conv_halting},C\ref{condition:non-negative},C\ref{condition:sum-to-1}\} };
		\node at (2, -2) (124) { \{C\ref{condition:conv_halting},C\ref{condition:non-negative},C\ref{condition:one-type}\} };
		\node at (4, -2) (134) { \{C\ref{condition:conv_halting},C\ref{condition:sum-to-1},C\ref{condition:one-type}\} };
		\node at (6, -2) (234) { \{C\ref{condition:non-negative},C\ref{condition:sum-to-1},C\ref{condition:one-type}\} };
		\node at (-0.4, -1) (12) { \{C\ref{condition:conv_halting},C\ref{condition:non-negative}\} };
		\node at (0.9, -1) (13) { \{C\ref{condition:conv_halting},C\ref{condition:sum-to-1}\} };
		\node at (2.3, -1) (14) { \{C\ref{condition:conv_halting},C\ref{condition:one-type}\} };
		\node at (3.7, -1) (23) { \{C\ref{condition:non-negative},C\ref{condition:sum-to-1}\} };
		\node at (5.1, -1) (24) { \{C\ref{condition:non-negative},C\ref{condition:one-type}\} };
		\node at (6.4, -1) (34) { \{C\ref{condition:sum-to-1},C\ref{condition:one-type}\} };
		\node at (0, 0) (1) { \{C\ref{condition:conv_halting}\} };
		\node at (2, 0) (2) { \{C\ref{condition:non-negative}\} };
		\node at (4, 0) (3) { \{C\ref{condition:sum-to-1}\} };
		\node at (6, 0) (4) { \{C\ref{condition:one-type}\} };
		\node at (3, 1) (empty) { $\emptyset$ };

		\draw (1234.north) -- (123.south);
		\draw (1234.north) -- (124.south);
		\draw (1234.north) -- (134.south);
		\draw (1234.north) -- (234.south);
		\draw (123.north) -- (12.south);
		\draw (123.north) -- (13.south);
		\draw (123.north) -- (23.south);
		\draw (124.north) -- (12.south);
		\draw (124.north) -- (14.south);
		\draw (124.north) -- (24.south);
		\draw (134.north) -- (13.south);
		\draw (134.north) -- (14.south);
		\draw (134.north) -- (34.south);
		\draw (234.north) -- (23.south);
		\draw (234.north) -- (24.south);
		\draw (234.north) -- (34.south);
		\draw (12.north) -- (1.south);
		\draw (12.north) -- (2.south);
		\draw (13.north) -- (1.south);
		\draw (13.north) -- (3.south);
		\draw (14.north) -- (1.south);
		\draw (14.north) -- (4.south);
		\draw (23.north) -- (2.south);
		\draw (23.north) -- (3.south);
		\draw (24.north) -- (2.south);
		\draw (24.north) -- (4.south);
		\draw (34.north) -- (3.south);
		\draw (34.north) -- (4.south);
		\draw (1.north) -- (empty.south);
		\draw (2.north) -- (empty.south);
		\draw (3.north) -- (empty.south);
		\draw (4.north) -- (empty.south);
	\end{tikzpicture}
	\caption{Lattice structure of the LEMM decision problem under different conditions.}
	\label{Figure: lattice}
\end{figure}
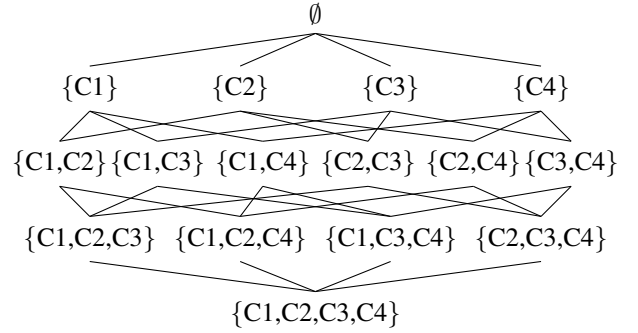

\begin{lemma}
	\label{lemma:np}
	The LEMM decision problem without any restriction is in \textsc{NP}.
\end{lemma}

\begin{proof}
	For all LEMM, index $i \in [n]$, and threshold $\beta \in \R$, if there is a feasible solution $\xx$ \st $x_m < \beta$, then $\xx$ is a certificate that can be verified in polynomial time.
\end{proof}

\begin{lemma}
	\label{lemma:linear system existence&uniqueness}
	Let $\mQ \in \R ^{n \times n}$. If $\lim _{m \rightarrow \infty} \mQ^m = \0 _{n \times n}$, then $(\mI - \mQ)$ is invertible. Further, if $\mQ \ge 0$, then $(\mI - \mQ)^{-1} \ge 0$.
\end{lemma}

\begin{remark}
	\citet{condon1992complexity} proves a version of this lemma under conditions~C\ref{condition:conv_halting},~C\ref{condition:non-negative}~and~C\ref{condition:sum-to-1}, and we observe that the lemma generalizes without any restrictions. {The proof is in the Appendix.}
\end{remark}

\begin{proposition}[Complexity results from the literature]
	\label{thm:ssg complexity}
	The following assertions hold:
	\begin{enumerate}
		\item \label{previous result: C2,C3 is in NP cap coNP} The LEMM decision problem under $\{\textnormal{C\ref{condition:non-negative},C\ref{condition:sum-to-1}}\}$ is in \textsc{UP} $\cap$ \textsc{coUP}.
		\item \label{previous result: C1,C2,C3=C2,C3} The LEMM decision problems under $\{\textnormal{C\ref{condition:conv_halting},C\ref{condition:non-negative},C\ref{condition:sum-to-1}}\}$ and under $\{\textnormal{C\ref{condition:non-negative},C\ref{condition:sum-to-1}}\}$ are polynomially equivalent.
		\item \label{previous result: C2,C3,C4 is in P} The LEMM decision problem under $\{\textnormal{C\ref{condition:non-negative},C\ref{condition:sum-to-1},C\ref{condition:one-type}}\}$ is in \textsc{PTIME}.
	\end{enumerate}
\end{proposition}

\noindent{\em Explanation.}
Items~\ref{previous result: C2,C3 is in NP cap coNP} and~\ref{previous result: C1,C2,C3=C2,C3} follow mainly from~\cite{condon1992complexity} and \cite{chatterjee2011reduction}.
Item~\ref{previous result: C2,C3,C4 is in P} follows from single-player stochastic games being MDPs, which can be solved via linear programming.

\subsection{Complexity lower bounds}
\label{subsec:lb}

In this section, we present the complexity lower bounds.
First, we recall the classic \textsc{NP}-complete partition problem.

\begin{lemma}[\citet{garey1979computers}]
	The partition problem is: given a set of positive integers $\{ a_1, \ldots, a_{k} \} \in \mathbb Z _{\ge 1} ^{k}$, determine whether the set can be partitioned into two subsets with equal sums.
	This problem is \textsc{NP}-complete.
\end{lemma}

\begin{lemma}
	\label{lemma:hardness of non-negative}
	The LEMM decision problem under $\{\textnormal{C\ref{condition:non-negative},C\ref{condition:one-type}}\}$ is \textsc{NP}-hard.
\end{lemma}
\begin{proof}[Main proof idea]
	We show a reduction from the partition problem of $\{ a_1, \ldots, a_{m} \}$ to an LEMM as follows.
	First, we construct a min-only LEMM with non-negative coefficients, in which the min variables can only take value $\pm 1$.
	Then, we associate each integer $a_i$ with a min variable $x_i$.
	When $x_i$ takes value $+1$, we put $a_i$ into the first subset; and when it takes value $-1$, we put $a_i$ into the second subset.
	Finally, we use the last affine variable $x_{2n_1+2}$ to encode the constraint that the sum of the two subsets is equal.
\end{proof}

\begin{lemma}
	\label{lemma:hardness of sum-to-one}
	The LEMM decision problem under $\{\textnormal{C\ref{condition:sum-to-1},C\ref{condition:one-type}}\}$ is \textsc{NP}-hard.
\end{lemma}

\begin{proof}[Main proof idea]
	From \cref{lemma:hardness of non-negative} and \cref{basic fact}, we know that the LEMM decision problem under $\{\textnormal{C\ref{condition:one-type}}\}$ is \textsc{NP}-hard. Then, we show that from the LEMM under $\{\textnormal{C\ref{condition:one-type}}\}$, we can construct another equivalent LEMM that satisfies conditions~C\ref{condition:sum-to-1}~and~C\ref{condition:one-type}.
\end{proof}

We conclude this section with a theorem.

\begin{theorem}
	\label{thm:np complete results}
	The LEMM decision problems under conditions $\{ \textnormal{C\ref{condition:non-negative},C\ref{condition:one-type}} \}$, $\{ \textnormal{C\ref{condition:sum-to-1},C\ref{condition:one-type}} \}$, $\{ \textnormal{C\ref{condition:non-negative}}\}$, $\{ \textnormal{C\ref{condition:sum-to-1}} \}$, $\{ \textnormal{C\ref{condition:one-type}} \}$, or  $\emptyset$ are all \textsc{NP}-complete.
	Hence, we establish the first row of \cref{tab:complexity results}.
\end{theorem}

\begin{proof}
	The \textsc{NP}-upper bound follows from Lemma~\ref{lemma:np} and \cref{basic fact}.
	The \textsc{NP}-hardness follows from Lemma~\ref{lemma:hardness of non-negative}, Lemma~\ref{lemma:hardness of sum-to-one}, and \cref{basic fact}.
\end{proof}

\subsection{Complexity upper bounds}
\label{subsec:ub}

In this section, we will establish the following main results.

\begin{theorem}
	\label{thm:halting is NP intersects co-NP}
	The LEMM decision problem under $\{\textnormal{C\ref{condition:conv_halting}}\}$ is in \textsc{UP} $\cap$ \textsc{coUP}.
\end{theorem}

\begin{corollary}
	\label{cor:up cap coup results}
	The \textsc{UP} $\cap$ \textsc{coUP} complexity results in the second, third, and fourth row of \cref{tab:complexity results} are established.
\end{corollary}

\noindent{\em Explanation.} \Cref{thm:ssg complexity} (item~\ref{previous result: C2,C3 is in NP cap coNP} and item~\ref{previous result: C1,C2,C3=C2,C3}) establishes the results of the fourth row of \cref{tab:complexity results}. Theorem~\ref{thm:halting is NP intersects co-NP} and \cref{basic fact} establish the results for the second and third row.

\begin{theorem}
	\label{thm:1,2,4 is in P}
	The LEMM decision problem under $\{\textnormal{C\ref{condition:conv_halting},C\ref{condition:non-negative},C\ref{condition:one-type}}\}$ is in \textsc{PTIME}.
\end{theorem}

\begin{corollary}
	\label{cor:PTIME results}
	The \textsc{PTIME} complexity results in the fifth row of \cref{tab:complexity results} are established.
\end{corollary}

\noindent{\em Explanation.}
\cref{thm:ssg complexity} (item~\ref{previous result: C2,C3,C4 is in P}) and \cref{basic fact} establishes the results for the LEMM decision problems under $\{ \textnormal{C\ref{condition:conv_halting},C\ref{condition:non-negative},C\ref{condition:sum-to-1},C\ref{condition:one-type}} \}$ and $\{ \textnormal{C\ref{condition:non-negative},C\ref{condition:sum-to-1},C\ref{condition:one-type}} \}$, and Theorem~\ref{thm:1,2,4 is in P} establishes the result for LEMMs under $\{ \textnormal{C\ref{condition:conv_halting},C\ref{condition:non-negative},C\ref{condition:one-type}} \}$.

\noindent\textbf{Classic properties.}
Before presenting our proofs, we explain the key difficulties in establishing the results as compared to previous analysis in the literature (\eg, \citet{condon1992complexity,ludwig1995subexponential,chatterjee2023faster}).
For halting stochastic games (\ie, the LEMM under $\{\textnormal{C\ref{condition:conv_halting},C\ref{condition:non-negative},C\ref{condition:sum-to-1}}\}$), the ``monotonicity'', ``minimax equality'', and ``solution from subsolution'' properties defined below hold and were used in previous analysis.
We show that these key properties do not hold under condition~C\ref{condition:conv_halting} only.
First, we define these properties.

\begin{itemize}
	\item
		\textbf{Monotonicity.}
		For LEMMs $(n_1,n_2,n,\cJ,\qq,\cdot)$ with a unique solution, we say that the LEMMs are monotonic if, for all $\bb_1 \ge \bb_2$, if $\xx^1$ is the solution of the LEMM with $(n_1,n_2,n,\cJ,\qq,\bb_1)$, and $\xx^2$ is the solution of the LEMM with $(n_1,n_2,n,\cJ,\qq,\bb_2)$, then $\xx^1 \ge \xx^2$.
	\item
		\textbf{Minimax equality.}
		Consider LEMMs $(n_1,n_2,n,\cJ,\qq,\bb)$ such that, for all assignments $\ell \colon [n_1+n_2] \mapsto [n]$ where
		$$
		\ell (k) \in \cJ(k), \ \forall k \in [n_1+n_2] \,,
		$$
		defining
		$$
		\mQ ^{(\ell)}
		= \left[ \ee_{\ell(1)}, \ldots, \ee_{\ell(n_1+n_2)}, \qq_{n_1+n_2+1}, \ldots, \qq_{n} \right] \T ,
		$$
		we have that $(\mI-\mQ^{(\ell)})$ is invertible.
		Then, taking
		$
		\xx^{(\ell)} = (\mI-\mQ^{(\ell)})^{-1} \bb
		$, the LEMM satisfies the minimax equality if, for all $k \in [n]$,
		$$
		\min _{\substack{\ell(i) \in \cJ (i), \\ 1 \le i \le n_1}} \max _{\substack{\ell(j) \in \cJ(j), \\ n_1 < j \le n_1+n_2}} x^{(\ell)}_k = \max _{\substack{\ell(j) \in \cJ(j), \\ n_1 < j \le n_1+n_2}} \min _{\substack{\ell(i) \in \cJ (i), \\ 1 \le i \le n_1}} x^{(\ell)}_k .
		$$
	\item
		\textbf{Solution from subsolution.}
		Consider LEMMs $(n_1,n_2,n,\cJ,\qq,\bb)$ and $1 \le i \le n_1$ (resp.\ $n_1 < j \le n_1 + n_2$), where, for all $l \in \cJ(i)$ (resp.\ $l \in \cJ(j)$), modifying the $i$-th equation in the LEMM by
		$$
		x_i = x_l \text{ (or $x_j = x_l$)}
		$$
		we obtain a new LEMM with a unique solution $\xx^{l}$ and let
		$$
		l^* = \arg \min _{l \in \cJ(i)} x ^{l} _i \text{ (or $ l^* = \arg \max _{l \in \cJ(j)} x ^{l} _j$)}.
		$$
		Then, we say the original LEMM has a solution from its subsolution if, for all $l \in \cJ(i)$ (resp.\ $l \in \cJ(j)$), we have that $\xx^{l^*}$ is the solution of the original LEMM.
\end{itemize}

We present illuminating counterexamples showing that the above classic properties break when condition C1 holds, and conditions C2 and C3 do not hold.

\begin{example}[Non-monotonicity]
	\label{example:non-monotonicity}
	Consider the LEMM
	$$
	\left\{
	\begin{aligned}
		x_1 &= \max\ \{ x_2 ,\ x_3 \} , \\
		x_2 &= -x_3 + 1, \\
		x_3 &= 0.5 x_3 + 0.2,
	\end{aligned}
	\right.
	$$
	which satisfies condition~C\ref{condition:conv_halting} and has a unique solution $\xx = [0.6, 0.6, 0.4] \T$.
	Increasing the value of $b_3$, we have the following LEMM
	$$
	\left\{
	\begin{aligned}
		x^\prime_1 &= \max\ \{ x^\prime_2 ,\ x^\prime_3 \} \,, \\
		x^\prime_2 &= -x^\prime_3 + 1 \,, \\
		x^\prime_3 &= 0.5 x^\prime_3 + 0.25 \,,
	\end{aligned}
	\right.
	$$
	with a unique solution $\xx^\prime = [0.5, 0.5, 0.5] \T$, so these LEMMs do not satisfy the monotonicity property.
	\qed
\end{example}

\begin{example}[$\min\ \max \neq \max\ \min$]
	\label{example:minimax inequality}
	Consider the LEMM
	$$
	\left\{
	\begin{aligned}
		x_1 &= \min \ \{ x_3, x_4 \} , \\
		x_2 &= \max \ \{ x_5, x_6 \} , \\
		x_3 &= -0.18 x_1 + 0.72 x_2 , \\
		x_4 &= 0.36 x_1 - 0.18 x_2 , \\
		x_5 &= 0.36 x_1 - 0.54 x_2 + 2 , \\
		x_6 &= -0.18 x_1 - 0.36 x_2 + 2 ,
	\end{aligned}
	\right.
	$$
	which satisfies condition~\textnormal{C\ref{condition:conv_halting}}, while we have
	$$
	\min _{l(1) \in \{3,4\}} \ \max _{l(2) \in \{5,6\}} \ x^{(l)}_2 = x^{(l(1)=3,\ l(2)=5)} _2 > 1.5 \,,
	$$
	and
	$$
	\max _{l(2) \in \{5,6\}} \ \min _{l(1) \in \{3,4\}} \ x^{(l)}_2 = x^{(l(1)=3,\ l(2)=6)} _2 < 1.4 \,.
	$$
	This violates the minimax equality.
	\qed
\end{example}

\begin{example}[Misleading solution from subsolution]
	\label{example:misleading subsolution}
	Consider the LEMM
	$$
	\left\{
	\begin{aligned}
		x_1 &= \min\ \{ x_2 ,\ x_3 \} \,, \\
		x_2 &= -0.1 x_1 + 0.8 x_4 + 2.2 \,, \\
		x_3 &= 0.1 x_1 + 0.5 x_4 + 2.2 \,, \\
		x_4 &= - 0.5 x_1 + 1.2 x_4 - 1.4 \,,
	\end{aligned}
	\right.
	$$
	which satisfies condition~C\ref{condition:conv_halting}.
	Choosing $x_1 = x_2$, the subsolution is
	$$
	\xx^{2}
	= \left[ -\frac {26} {3}, - \frac {26} {3}, -6, - \frac {44} {3} \right] \T,
	$$
	while choosing $x_1 = x_3$, the subsolution is
	$$
	\xx^{3}
	= \left[ - \frac {114} {7}, - \frac {162} {7}, -\frac {114} {7}, -\frac {236} {7} \right] \T .
	$$
	Even when $x^{2}_1 > x^{3}_1$, the \emph{larger} subsolution $\xx^{2}$ is the solution of the original LEMM.
	This violates the solution from subsolution property, which states that $x^{3}$ should be the solution to the original LEMM.
	\qed
\end{example}

\noindent{\bf Novelty.}
We present our key technical lemma, which shows that condition~C\ref{condition:conv_halting} already ensures the existence of a unique solution.
In the context of halting SSGs (\ie, $\{\textnormal{C\ref{condition:conv_halting},C\ref{condition:non-negative},C\ref{condition:sum-to-1}}\}$), the proof of NP $\cap$ coNP is based on the minimax equality since, given the choices of the min (resp.\ max) variables as the certificate of the \textsc{Yes} (resp.\ \textsc{No}) answer, it can be verified by solving the remaining max-only (resp.\ min-only) subproblem via linear programming \cite{condon1992complexity,ludwig1995subexponential,chatterjee2023faster}.
Our main novelty is to present a new analysis technique that establishes the desired complexity results even though the classic properties that were crucial in the previous analysis do not hold.

\begin{lemma}[Key lemma]
	\label{lemma:existence&uniqueness}
	Under condition~\textnormal{C\ref{condition:conv_halting}}, there exists a unique solution to the LEMM.
\end{lemma}

\begin{proof}[Main proof idea]
	We prove by induction on $n_1+n_2$.
	
	(i) \emph{Base case.} The base case $n_1+n_2=0$ follows from \cref{lemma:linear system existence&uniqueness}.
	
	(ii) \emph{Inductive case.} Assume by induction that the existence and uniqueness hold when $n_1 + n_2 = m$, and we consider the case $n_1+n_2=m+1$.
	Without loss of generality, let $x_{m+1} = \max \{x_i, x_j\}$.
	By contradiction, assume that the LEMM either has no solution or has $2$ different solutions.
	For the reduced LEMM obtained by fixing $x_{m+1}$ as $x_j$ (resp.\ $x_i$), we get its unique solution $\xx^{(0)}$ (resp.\ $\xx^{(1)}$) by the induction hypothesis.
	Then, for any $\alpha \in (0, 1)$, fixing $x_{m+1}$ as $\alpha x_i + (1-\alpha) x_j$, we get a unique solution $\xx^{(\alpha)}$ for the reduced LEMM by induction hypothesis.
	Moreover, by continuity, there is $\alpha^*$ such that $\xx^{(\alpha^*)}_i = \xx^{(\alpha^*)}_j$.
	This is a contradiction because $\xx^{(\alpha^*)}$ becomes the unique solution for the original LEMM.
\end{proof}

\begin{proof}
	Without loss of generality, assume that $|\cJ(i)| = 2$, for all $i \in [n_1+n_2]$.
	We proceed by induction on $(n_1 + n_2)$.
	
	(i) \emph{Base case.}
	For $n_1 + n_2 = 0$, by Lemma~\ref{lemma:linear system existence&uniqueness}, there exists a unique solution $\xx = (\mI - \mQ)^{-1} \bb$, where $\mQ$ is the only element in $\cQ$.
	
	(ii) \emph{Inductive case.}
	By induction, all LEMMs with $n_1 + n_2 = m$ have a unique solution.
	We prove it is the case for all LEMMs with $n_1 + n_2 = m + 1$.
	By contradiction, consider an LEMM $(n_1,n_2,n,\cJ,\qq,\bb)$ with $n_1+n_2 = m + 1$ that has no solution or has $2$ different solutions.
	Without loss of generality $\cJ(m+1) = \{ i, j \}$, and $x_{m+1} = \max \{ x_i, x_j \}$ (if $x_{m+1}$ is a min variable, we only switch some inequalities). For all $\alpha \in [0,1]$, consider replacing this equation by $x_{m+1} = \alpha\ x_i + (1 - \alpha)\ x_j$.
	Then, by the induction hypothesis, the modified LEMM has a unique solution denoted $\xx^{(\alpha)}$.
	
	Note that $\xx^{(0)}$ is a solution to the original system with $(m+1)$ min and max variables if and only if ${x^{(0)}}_i \ge {x^{(0)}}_j$.
	Similarly, $\xx^{(1)}$ is a solution to the original system if and only if ${x^{(1)}}_i \le {x^{(1)}}_j$.
	Since the LEMM has no solution or has $2$ different solutions, the following holds:
	\begin{equation}
		\label{eq:ij<0}
		\left( {x^{(0)}}_i - {x^{(0)}}_j \right) \cdot \left( {x^{(1)}}_i - {x^{(1)}}_j \right) < 0 \,.
	\end{equation}
	
	We introduce the following notations for $\alpha \in [0,1]$.
	\begin{align*}
		\mE
		&\coloneqq [\ee_1, \cdots, \ee_{m}, \0, \ee_{m+2}, \cdots, \ee_n ] \T \,, \\
		\cQ ^{(\alpha)}
		&\coloneqq \left\{ \mE \mQ + \ee_{m+1} \left( \alpha \ee_i + (1-\alpha) \ee_j \right) \T \;\middle|\; \mQ \in \cQ \right\} \,, \\
		\cM ^{(\alpha)}
		&\coloneqq\left\{ \mM \in \cQ ^{(\alpha)} \;\middle|\; \xx^{(\alpha)} = \mM \xx^{(\alpha)} + \bb \right\} \,.
	\end{align*}
	Intuitively, $\mE \mQ + \ee_{m+1} \left( \alpha \ee_i + (1-\alpha) \ee_j \right) \T $ corresponds to replacing the $(m+1)$-th row of $\mQ$ by $(\alpha \ee_i + (1-\alpha)\ee_j) \T$.
	
	We show that $\alpha \in [0, 1] \mapsto \xx^{(\alpha)}$ is continuous.
	By the continuity of linear, $\min$, and $\max$ operators, there exists $\epsilon > 0$, such that, for all $\alpha^\prime \in \cB(\alpha; \epsilon) \coloneqq (\alpha-\epsilon, \alpha+\epsilon) \cap [0,1]$, defining $\cM ^{(\alpha; \alpha^\prime)}$ by
	\[
	\left\{ \mE \mM + \ee_{m+1} \left( \alpha^\prime \ee_i + (1-\alpha^\prime) \ee_j \right) \T \;\middle|\; \mM \in \cM^{(\alpha)} \right\} \,,
	\]
	we have
	$$
	\cM ^{(\alpha^\prime)}
	\subseteq \cM ^{(\alpha; \alpha^\prime)} \,.
	$$
	Therefore, for all sequences $(\alpha_n) _{n \in \N} \subseteq \cB(\alpha;\epsilon)$ that converges to $\alpha$, we have that
	\[
	\lvert \xx ^{(\alpha_n)} - \xx ^{(\alpha)} \rvert = \sup _{\mM \in \cM^{(\alpha_n)}} \lvert (\mI - \mM) ^{-1} \bb - \xx^{(\alpha)} \rvert
	\]
	is bounded from above by
	\[
	\sup _{\mM \in \cM^{(\alpha; \alpha_n)}} \lvert (\mI - \mM) ^{-1} \bb - \xx^{(\alpha)} \rvert \,,
	\]
	which is equal to the supremum over $\mM \in \cM^{(\alpha)}$ of
	\[
	\left|
	\begin{aligned}
		&\left( \mI - \mE \mM - \ee_{m+1} \left( \alpha_n \ee_i + (1-\alpha_n) \ee_j \right) \T \right) ^{-1} \bb \\
		&\quad - \left( \mI - \mE \mM - \ee_{m+1} \left( \alpha \ee_i + (1-\alpha) \ee_j \right) \T \right) ^{-1} \bb
	\end{aligned}
	\right| \,,
	\]
	which converges to $0$ as $n$ grows.
	Hence, $\alpha \in [0, 1] \mapsto \xx^{(\alpha)}$ is continuous.
	
	Going back to \cref{eq:ij<0}, by continuity, there exists $\alpha^\star \in [0,1]$ such that
	$$
	{x^{(\alpha^\star)}}_i - {x^{(\alpha^\star)}}_j = 0\,.
	$$
	Therefore, the two LEMMs after the modification $x_{m+1} = x_i$ and $x_{m+1} = x_j$ have the same solution $\xx^{(\alpha^*)}$, \ie, $\xx^{(0)} = \xx^{(1)} = \xx^{(\alpha^*)}$, so the original LEMM with $(m + 1)$ min and max variables has a unique solution $\xx^{(\alpha^\star)}$.
	This yields the desired contradiction, which concludes the proof by induction.
\end{proof}

Finally, we use \Cref{lemma:existence&uniqueness} to prove Theorem~\ref{thm:halting is NP intersects co-NP} and Theorem~\ref{thm:1,2,4 is in P}.

\begin{proof}[Proof of \cref{thm:halting is NP intersects co-NP}]
	By \Cref{lemma:existence&uniqueness}, we can use the unique solution of the LEMM as a (unique) certificate of the answer \textsc{Yes} or of the answer \textsc{No}, which is verifiable in polynomial time.
\end{proof}

\begin{proof}[Main proof idea of Theorem~\ref{thm:1,2,4 is in P}]
	Consider an LEMM $(n_1,n_2,n,\cJ,\qq,\bb)$ satisfying conditions~C\ref{condition:conv_halting},~C\ref{condition:non-negative}~and~C\ref{condition:one-type}.
	By \Cref{lemma:existence&uniqueness}, there exists a unique solution.
	Without loss of generality, assume it only has max-variables.
	Then, the solution is given by solving the linear program:
	$$
	\begin{alignedat}{2}
		\min \ &x_1 + \ldots + x_{n} \\
		\st &x_i \ge x_l,\quad & \forall l \in \cJ(i),\ \forall 1 \le i \le n_1+n_2 \,, \\
		&x_k \ge \qq_k\T \xx + b_k,\quad & \forall n_1+n_2 < k \le n \,.
	\end{alignedat}
	$$
	The desired result follows.
\end{proof}

\subsection{Equivalence between sub-classes of problems}
\label{subsec:equiv}

We establish the equivalence between the LEMM decision problems in the second row of \cref{tab:complexity results}.
\begin{lemma}
	\label{lemma:equivalence c1=c134}
	There is a polynomial time reduction from
	LEMMs under $\{\textnormal{C\ref{condition:conv_halting}}\}$ to LEMMs under $\{\textnormal{C\ref{condition:conv_halting},C\ref{condition:sum-to-1},C\ref{condition:one-type}}\}$.
\end{lemma}

\begin{proof}[Main proof idea]
	From LEMMs under condition~C\ref{condition:conv_halting}, we construct a new, equivalent, min-only LEMM (that satisfies conditions~C\ref{condition:conv_halting}~and~C\ref{condition:one-type}).
	To do so, we introduce negative copies of the variables, so the $\max$ equations can be rewritten as $\min$ equations over the negative copies.
	We note that extra positive copies are required to keep the periodicity and thus the halting condition. Finally, we introduce a dummy variable to ensure condition~C\ref{condition:sum-to-1}.
\end{proof}

\begin{corollary}
	\label{cor:equivalence and ssg hardness}
	The LEMM decision problems under $\{\textnormal{C\ref{condition:conv_halting},C\ref{condition:sum-to-1},C\ref{condition:one-type}}\}$, $\{\textnormal{C\ref{condition:conv_halting},C\ref{condition:sum-to-1}}\}$, $\{\textnormal{C\ref{condition:conv_halting},C\ref{condition:one-type}}\}$, or $\{\textnormal{C\ref{condition:conv_halting}}\}$ are polynomially equivalent and SSG-hard.
	Hence, we establish the equivalence of the problems in the second row of \cref{tab:complexity results} and their SSG-hardness.
\end{corollary}
\begin{proof}
	The equivalence follows from \cref{lemma:equivalence c1=c134}.
	The SSG-hardness of $\{\textnormal{C\ref{condition:conv_halting}}\}$ follows from \cref{thm:ssg complexity} and \cref{basic fact}.
\end{proof}

\section{Complexity of condition decision problems}
\label{sec:checking conditions}

In this section, we establish complexity results for the condition decision problem.
In what follows, we consider that inputs of the problems are given as rational numbers as usual.
We start with a remark on the basic and previous complexity results and a basic fact.

\begin{lemma}[Basic and previous complexity results]
	The condition decision problems for $\{ \textnormal{C\ref{condition:non-negative}} \}$, $\{ \textnormal{C\ref{condition:sum-to-1}} \}$, or $\{ \textnormal{C\ref{condition:one-type}} \}$ are all easily solvable in linear time. Moreover, the condition decision problem for $\{ \textnormal{C\ref{condition:conv_halting},C\ref{condition:non-negative},C\ref{condition:sum-to-1}} \}$ is in \textsc{PTIME} \citep{baier2008principles}.
\end{lemma}

\begin{fact}
	\label{fact:checking}
	If $\{\textnormal{C\ref{condition:conv_halting}}\} \subseteq X \subseteq Y$ represents two subsets of conditions, then the condition decision problem for $X$ is no easier than the condition decision problem for $Y$. Hence any complexity lower bound for $Y$ also holds for $X$, and any complexity upper bound for $X$ also holds for $Y$.
\end{fact}

We present the following characterization, which implies that checking conditions C\ref{condition:conv_halting} and C\ref{condition:non-negative} can be done in polynomial time via linear programming.

\begin{lemma}[Technical lemma]
	\label{lemma:checking via LP}
	Consider the LEMM with $(n_1,n_2,n,\cJ,\qq,\bb)$ such that $\qq_k \ge 0$, for all $k \in [n_1+n_2+1, n]$.
	Then, condition~C\ref{condition:conv_halting} holds if and only if
	$$
	\exists \xx \in \R_{\ge 0}^n \ \forall \mQ \in \cQ \quad \xx \ge \mQ \xx + \1 \,.
	$$
\end{lemma}

\begin{theorem}
	\label{thm:PTIME check}
	The condition decision problem for $\{\textnormal{C\ref{condition:conv_halting},C\ref{condition:non-negative}}\}$ is in \textsc{PTIME}.
\end{theorem}

\begin{proof}
	Condition~C\ref{condition:non-negative} can be easily checked.
	By \Cref{lemma:checking via LP}, under condition~C\ref{condition:non-negative},
	we can check condition~C\ref{condition:conv_halting} via the feasibility of the following system of linear inequalities:
	$$
	\left\{
	\begin{aligned}
		x_i &\ge x_l + 1, \quad & 1 \le i \le n_1+n_2, \ l \in \cJ(i) , \\
		x_k &\ge \qq_k \T \xx + 1, \quad & n_1+n_2  < k \le n, \\
		\xx &\ge 0 . 
	\end{aligned}
	\right.
	$$
\end{proof}

\begin{corollary}
	\label{cor:PTIME checks}
	The condition decision problems for {$\{\textnormal{C\ref{condition:conv_halting},C\ref{condition:non-negative}}\}$}, {$\{\textnormal{C\ref{condition:conv_halting},C\ref{condition:non-negative},C\ref{condition:one-type}}\}$}, $\{\textnormal{C\ref{condition:conv_halting},C\ref{condition:non-negative},C\ref{condition:sum-to-1}}\}$, or $\{\textnormal{C\ref{condition:conv_halting},C\ref{condition:non-negative},C\ref{condition:sum-to-1},C\ref{condition:one-type}}\}$ are in \textsc{PTIME}.
	Hence, we establish the second row of \cref{tab:condition decision}.
\end{corollary}
\begin{proof}
	It follows from \cref{thm:PTIME check} and \cref{fact:checking}.
\end{proof}

Finally, it remains the condition decision problem for condition~C\ref{condition:conv_halting}.
We show that this problem is \textsc{coNP}-hard, even under conditions~C\ref{condition:sum-to-1}~and~C\ref{condition:one-type}.

\begin{theorem}
	\label{thm:checking c1,3,4 is conp hard}
	The condition decision problem for $\{\textnormal{C\ref{condition:conv_halting},C\ref{condition:sum-to-1},C\ref{condition:one-type}}\}$ is \textsc{coNP}-hard.
\end{theorem}

\begin{proof}[Main proof idea]
	We show the reduction from the \textsc{SAT} problem in conjunctive normal form with variables $v_1, \ldots, v_r$ and clauses $c_1, \ldots, c_m$ to LEMM as follows.
	
	First, we construct the variables: (i) the max variables $x_1, \ldots, x_m$ associated with $c_1, \ldots, c_m$, (ii) the max variables $x_{m+1}, \ldots, x_{m+r}$ associated with $v_1, \ldots, v_r$, and their negative copies $x_{m+r+1}, \ldots, x_{m+2r}$, as well as (iii) an (affine) variable $x_{m+2r+1}$ and its negative copy $x_{m+2r+2}$.
	
	Then, we construct the transitions: (i) From $x_i$ ($1 \le i \le m$), we can go to $x_{m+l}$ if $v_l$ is a literal in $c_i$, or to the negative copy of $x_{m+l}$ if $\bar v_l$ is literal in $c_i$. (ii) From $x_j$ ($m \le j \le m+r$), we can go to $x_{m+2r+1}$ or its negative copy.
	
	Finally, we show that ``condition~C\ref{condition:conv_halting} does not hold'' if and only if ``from each $x_i$ ($i \in [m]$), we can reach $x_{m+2r+1}$ (or its negative copy) via some $x_{m+l}$ (or its negative copy)''. This corresponds to ``for each $c_i$ ($i \in [m]$), some assignment of $v_l$ makes $c_i$ true'', \ie, the original SAT instance is satisfiable.
\end{proof}

\begin{corollary}
	\label{cor:coNP-hard checks}
	The condition decision problems for {$\{\textnormal{C\ref{condition:conv_halting}}\}$}, {$\{\textnormal{C\ref{condition:conv_halting},C\ref{condition:sum-to-1}}\}$}, {$\{\textnormal{C\ref{condition:conv_halting},C\ref{condition:one-type}}\}$}, or {$\{\textnormal{C\ref{condition:conv_halting},C\ref{condition:sum-to-1},C\ref{condition:one-type}}\}$} are \textsc{coNP}-hard.
	Hence, we establish the first row of \cref{tab:condition decision}.
\end{corollary}
\begin{proof}
    This follows from Theorem~\ref{thm:checking c1,3,4 is conp hard} and \cref{fact:checking}.
\end{proof}

\section{Conclusions}

In this work, we consider LEMMs with various conditions. While previous results from the literature consider specific cases, our results establish the computational complexity landscape for all possible subsets of conditions and the complexity of checking the conditions.
Exploring our approach (\eg, with condition C\ref{condition:conv_halting} only, or with conditions C\ref{condition:conv_halting} and C\ref{condition:non-negative}) in practical applications is an interesting direction for future work.

\section*{Acknowledgments}

This research was partially supported by the ERC CoG 863818 (ForM-SMArt) grant and the Austrian Science Fund (FWF) 10.55776/COE12 grant.

\bibliography{reference}

\begin{thebibliography}{20}
\providecommand{\natexlab}[1]{#1}

\bibitem[{Athreya and Ney(2004)}]{athreya2004branching}
Athreya, K.~B.; and Ney, P.~E. 2004.
\newblock \emph{Branching processes}.
\newblock Courier Corporation.

\bibitem[{Baier and Katoen(2008)}]{baier2008principles}
Baier, C.; and Katoen, J.-P. 2008.
\newblock \emph{Principles of model checking}.
\newblock MIT press.

\bibitem[{Bottou(2010)}]{bottou2010large}
Bottou, L. 2010.
\newblock Large-scale machine learning with stochastic gradient descent.
\newblock In \emph{Proceedings of the 19th International Conference on
  Computational Statistics (COMPSTAT)}, 177--186.

\bibitem[{Chatterjee and Fijalkow(2011)}]{chatterjee2011reduction}
Chatterjee, K.; and Fijalkow, N. 2011.
\newblock A reduction from parity games to simple stochastic games.
\newblock In \emph{Proceedings of the Second International Symposium on Games,
  Automata, Logics and Formal Verification (GandALF)}, volume~54 of
  \emph{{EPTCS}}, 74--86.

\bibitem[{Chatterjee et~al.(2024)Chatterjee, Kafshdar~Goharshady, Karrabi,
  Novotny, and {\v Z}ikeli{\'c}}]{chatterjee2024solving}
Chatterjee, K.; Kafshdar~Goharshady, E.; Karrabi, M.; Novotny, P.; and {\v
  Z}ikeli{\'c}, D. 2024.
\newblock {Solving Long-run Average Reward Robust MDPs via Stochastic Games}.
\newblock In \emph{Proceedings of the Thirty-Third International Joint
  Conference on Artificial Intelligence ({{IJCAI-24}})}, 6707--6715.

\bibitem[{Chatterjee et~al.(2023)Chatterjee, Meggendorfer, Saona, and
  Svoboda}]{chatterjee2023faster}
Chatterjee, K.; Meggendorfer, T.; Saona, R.; and Svoboda, J. 2023.
\newblock Faster algorithm for turn-based stochastic games with bounded
  treewidth.
\newblock In \emph{Proceedings of the Annual ACM-SIAM Symposium on Discrete
  Algorithms (SODA)}, 4590--4605.

\bibitem[{Condon(1990)}]{condon1990algorithms}
Condon, A. 1990.
\newblock On Algorithms for Simple Stochastic Games.
\newblock \emph{Advances in computational complexity theory}, 13: 51--72.

\bibitem[{Condon(1992)}]{condon1992complexity}
Condon, A. 1992.
\newblock The Complexity of Stochastic Games.
\newblock \emph{Information and Computation}, 96(2): 203--224.

\bibitem[{Dantzig(2002)}]{dantzig2002linear}
Dantzig, G.~B. 2002.
\newblock Linear programming.
\newblock \emph{Operations research}, 50(1): 42--47.

\bibitem[{Garey and Johnson(1979)}]{garey1979computers}
Garey, M.~R.; and Johnson, D.~S. 1979.
\newblock \emph{Computers and intractability}, volume 174.
\newblock freeman San Francisco.

\bibitem[{G{\"a}rtner and Matousek(2012)}]{gartner2012approximation}
G{\"a}rtner, B.; and Matousek, J. 2012.
\newblock \emph{Approximation algorithms and semidefinite programming}.
\newblock Springer Science \& Business Media.

\bibitem[{Goodfellow et~al.(2013)Goodfellow, Warde-Farley, Mirza, Courville,
  and Bengio}]{goodfellow2013maxout}
Goodfellow, I.; Warde-Farley, D.; Mirza, M.; Courville, A.; and Bengio, Y.
  2013.
\newblock Maxout networks.
\newblock In \emph{International conference on machine learning}, 1319--1327.

\bibitem[{Hemaspaandra and Rothe(1997)}]{hemaspaandra1997unambiguous}
Hemaspaandra, L.~A.; and Rothe, J. 1997.
\newblock Unambiguous computation: Boolean hierarchies and sparse
  Turing-complete sets.
\newblock \emph{SIAM Journal on Computing}, 26(3): 634--653.

\bibitem[{Katz et~al.(2017)Katz, Barrett, Dill, Julian, and
  Kochenderfer}]{katz2017reluplex}
Katz, G.; Barrett, C.; Dill, D.~L.; Julian, K.; and Kochenderfer, M.~J. 2017.
\newblock Reluplex: An efficient SMT solver for verifying deep neural networks.
\newblock In \emph{Proceedings of the 29th International Conference on Computer
  Aided Verification (CAV)}, 97--117.

\bibitem[{Ludwig(1995)}]{ludwig1995subexponential}
Ludwig, W. 1995.
\newblock A subexponential randomized algorithm for the simple stochastic game
  problem.
\newblock \emph{Information and computation}, 117(1): 151--155.

\bibitem[{Markowitz(1952)}]{markowitz1952PortfolioSelection}
Markowitz, H. 1952.
\newblock Portfolio {{Selection}}.
\newblock \emph{The Journal of Finance}, 7(1): 77.

\bibitem[{Nesterov(2013)}]{nesterov2013introductory}
Nesterov, Y. 2013.
\newblock \emph{Introductory lectures on convex optimization: A basic course},
  volume~87.
\newblock Springer Science \& Business Media.

\bibitem[{Puterman(2014)}]{puterman2014markov}
Puterman, M.~L. 2014.
\newblock \emph{Markov {{Decision Processes}}: {{Discrete Stochastic Dynamic
  Programming}}}.
\newblock Wiley.

\bibitem[{Shapley(1953)}]{shapley1953stochastic}
Shapley, L.~S. 1953.
\newblock Stochastic {{Games}}.
\newblock \emph{Proceedings of the National Academy of Sciences}, 39(10):
  1095--1100.

\bibitem[{von Neumann and Morgenstern(1953)}]{von1953theory}
von Neumann, J.; and Morgenstern, O. 1953.
\newblock \emph{Theory of games and economic behavior}.
\newblock Princeton university press.

\end{thebibliography}

\clearpage
\appendix


\section{Appendix: Details and proofs}
\label{sec:proof details}

\subsection{Deferred proofs in \cref{subsec:basic results}}

\begin{proof}[Proof of \cref{lemma:linear system existence&uniqueness}]
	If $(\mI-\mQ)$ is not invertible, then $1$ is an eigenvalue of $\mQ$, and the spectral radius of $\mQ$ is at least $1$, which leads to a contradiction. Therefore, $(\mI-\mQ)$ is invertible.
	Further, since
	$$
	\begin{aligned}
		\mI &= \lim _{m \rightarrow \infty} \left( \mI - \mQ^m \right) \\
		&= (\mI-\mQ) \cdot \lim _{m \rightarrow \infty} \left( \mI+\mQ+\ldots+\mQ^{m-1} \right) ,
	\end{aligned}
	$$
	if $\mQ \ge 0$, we have
	$$
	(\mI-\mQ)^{-1} = \lim _{m \rightarrow \infty} \left( \mI+\mQ+\ldots+\mQ^{m-1} \right) \ge 0 .
	$$
\end{proof}

\subsection{Deferred proofs in \cref{subsec:lb}}

\begin{proof}[Proof of Lemma~\ref{lemma:hardness of non-negative}]
	For all set of natural numbers $\{ a_1, \ldots, a_{m} \} \in \N ^{m}$, we define
	the LEMM with $(n_1,n_2,n,\cJ,\qq,\bb)$ where, for all $1 \le i \le m$ and $m  < k \le 2m$,
	$$
	\begin{alignedat}{2}
		n_1 &= m, \\
		n_2 &= 0, \\
		n &= 2m+1, \\
		\cJ(i) &= \{ i+m, 2m+1 \}, \\
		\qq_k &= [ 2\delta_{(k-m),1} , \ldots, 2\delta_{(k-m),m}, 0, \ldots, 0] \T , \\
		\qq_{2m+1} &= \0 , \\
		\qq_{2m+2} &= [a_1, \ldots, a_{m}, 0, \ldots, 0, 1] \T , \\
		b_{m+1} &= \ldots = b_{2m+1} = 1, \\
		b_{2m+2} &= 0 ,
	\end{alignedat}
	$$
	or equivalently,
	\begin{equation}
		\label{eq:reduction to partition problem}
		\left\{
		\begin{aligned}
			x_i &= \min\ \{ x_{i+n_1} , x_{2n_1+1} \} , \hspace{3em} 1 \le i \le n_1, \\
			x_j &= 2 x_{j-n_1} + 1 , \hspace{6em} n_1 < j \le 2n_1, \\
			x_{2n_1+1} &= 1, \\
			x_{2n_1+2} &= a_1 x_1 + \ldots + a_{n_1} x_{n_1} + x_{2n_1+2} .
		\end{aligned}
		\right.
	\end{equation}
	Note that conditions~C\ref{condition:non-negative}~and~C\ref{condition:one-type} are satisfied.
	To read \cref{eq:reduction to partition problem}, $x_i$ can only take value $1$ or $-1$, for all $ 1 \le i \le n_1$, and the last equation in \eqref{eq:reduction to partition problem} is essentially
	$$
	0 = a_1 x_1 + \ldots + a_{n_1} x_{n_1} .
	$$
	Therefore, the set of natural numbers $\{a_1, \ldots, a_{m}\}$ can be partitioned into two subsets of equal sums if and only if \cref{eq:reduction to partition problem} has a solution, which in turns is true if and only if the answer to the LEMM decision problem defined above is \textsc{Yes}.
	To conclude, \Cref{lemma:hardness of non-negative} follows from the \textsc{NP}-hardness of the partition problem.
\end{proof}

\begin{proof}[Proof of Lemma~\ref{lemma:hardness of sum-to-one}]
	We show a reduction from the LEMM decision problem under $\{\textnormal{C\ref{condition:one-type}}\}$ to the LEMM decision problem under $\{\textnormal{C\ref{condition:sum-to-1},C\ref{condition:one-type}}\}$.
	For all LEMMs $(n_1,n_2,n,\cJ,\qq,\bb)$ where $n_1 = 0$ or $n_2=0$, we define the following LEMM
	\begin{equation}
		\label{eq:reduction sum to 1}
		\begin{cases}
			x^\prime_i = \min\ \left\{ x^\prime_{l} \;\middle|\; {l \in \cJ (i)} \right\} , \quad & 1 \le i \le n_1, \\
			x^\prime_j = \max\ \left\{ x^\prime_{l} \;\middle|\; {l \in \cJ (j)} \right\} , \quad & n_1 < j \le n_1+n_2, \\
			x^\prime_k = \begin{bmatrix} \qq_k \\ - {\qq_k} \T \1 - b_k \end{bmatrix} \T \xx^\prime + b_k, \quad & n_1+n_2 < k \le n , \\
			x^\prime _{n+1} = 0 ,
		\end{cases}
	\end{equation}
	which satisfies conditions C\ref{condition:sum-to-1} and C\ref{condition:one-type}.
	Note that $\xx$ is a solution to the LEMM $(n_1,n_2,n,\cJ,\qq,\bb)$ if and only if $\xx^\prime = \begin{bmatrix} \xx \\ 0 \end{bmatrix}$ is a solution to \cref{eq:reduction sum to 1}.
	
	By \Cref{lemma:hardness of non-negative}, the LEMM decision problem under $\{\textnormal{C\ref{condition:one-type}}\}$ is \textsc{NP}-hard.
	Therefore, the \textsc{NP}-hardness of the LEMM decision problem under $\{\textnormal{C\ref{condition:sum-to-1},C\ref{condition:one-type}}\}$ follows from the above reduction.
\end{proof}

\subsection{Deferred proofs in \cref{subsec:ub}}

\begin{proof}[Proof of Theorem~\ref{thm:1,2,4 is in P}]
	Consider an LEMM $(n_1,n_2,n,\cJ,\qq,\bb)$ satisfying conditions~C\ref{condition:conv_halting},~C\ref{condition:non-negative}~and~C\ref{condition:one-type}.
	Without loss of generality, we assume $n_1 = 0$.
	
	First, we consider the following linear programming problem
	\begin{equation}
		\label{eq:LP for 124}
		\begin{alignedat}{2}
			\min \ &x_1 + \ldots + x_{n} \\
			\st &x_i \ge x_l,\quad & \forall l \in \cJ(i),\ \forall 1 \le i \le n_1+n_2 \,, \\
			&x_k \ge \qq_k\T \xx + b_k,\quad & \forall n_1+n_2 < k \le n \,.
		\end{alignedat}
	\end{equation}
	By \cref{lemma:existence&uniqueness}, the LEMM has a solution.
	Therefore, its solution is a feasible solution to the linear programming problem \eqref{eq:LP for 124}.
	Moreover,
	$$
	\xx \ge \mQ \xx + \bb, \ \forall \mQ \in \cQ \,.
	$$
	By \cref{lemma:linear system existence&uniqueness}, we have that
	$$
	\1 \T \xx = \1 \T (\mI - \mQ)^{-1} (\mI - \mQ) \xx \ge \1 \T (\mI - \mQ)^{-1} \bb \ge 0 \,.
	$$
	Therefore, the linear programming problem~\eqref{eq:LP for 124} has a finite optimal solution.
	
	Let $\xx^*$ be one of the optimal solutions, which can be obtained in polynomial time.
	We show that $\xx^*$ is the solution to the LEMM:
	\begin{itemize}
		\item
		For all $1 \le i \le n_1+n_2$, if $\xx^*_i > \max _{l \in \cJ (i)} \ \xx^*_l$, we can construct $\xx^\prime \in \R^n$ as follows:
		$$
		\xx^\prime_j = \begin{cases}
			\xx^*_j, & j \neq i \,, \\
			\max\ _{l \in \cJ (i)} \ \xx^*_l, & j = i \,.
		\end{cases}
		$$
		Then, $\xx^\prime$ is also feasible and it violates the optimality of $\xx^*$.
		Therefore, we must have $\xx^*_i = \max _{l \in \cJ (i)} \ \xx^*_l$.
		\item
		For all $n_1+n_2 < k \le n$, if $\xx^*_k > \qq_k \T \xx^* + b_k$, we construct $\xx^{\prime\prime} \in \R^n$ as follows:
		$$
		\xx^{\prime\prime}_j = \begin{cases}
			\xx^*_j, & j \neq k \,, \\
			\qq_k \T \xx^* + b_k, & j = k \,.
		\end{cases}
		$$
		Then, $\xx^{\prime\prime}$ is also feasible and it violates the optimality of $\xx^*$.
		Therefore, we must have $\xx^*_k = \qq_k \T \xx^* + b_k$.
	\end{itemize}
	
	By \Cref{lemma:existence&uniqueness}, the solution $\xx^*$ is unique to the LEMM. Hence we can answer the LEMM decision problem using this solution.
\end{proof}

\subsection{Deferred proofs in \cref{subsec:equiv}}

\begin{proof}[Proof of Lemma~\ref{lemma:equivalence c1=c134}]
	Consider an LEMM $(n_1, n_2, n, \cJ, \qq, \bb)$ satisfying condition~C\ref{condition:conv_halting}.
	Define, for $i \in [1,n_1]$,
	\begin{align*}
		\cJ ^\prime (i)
		&= \{ l+2n \mid l \in \cJ(i) \cap [1,n_1] \} \ \cup \\
		&\qquad \{ l+n \mid l \in \cJ(i) \cap [n_1+1,n] \}
	\end{align*}
	for $i \in [n_1+1, n_1+n_2]$,
	\begin{align*}
		\cJ ^\prime (i)
		&= \{ l+n \mid l \in \cJ(i) \cap [1,n_1] \} \ \cup \\
		&\quad \{ l+2n \mid l \in \cJ(i) \cap [n_1+1,n] \}
	\end{align*}
	and
	$$
	\bar \qq _k = \begin{bmatrix}
		\0_{n \times n_1} & \0_{n \times (n-n_1)} \\
		\mI _{n_1} & \0_{n_1 \times (n-n_1)} \\
		\0_{n \times n_1} & \0_{n \times (n-n_1)} \\
		\0_{(n-n_1) \times n_1} & \mI _{n-n_1}
	\end{bmatrix} \cdot \qq_k , \quad n_2 < k \le n.
	$$
	We consider the following LEMM: 
	\begin{equation}
		\label{eq:reduction to C 1,3,4}
		\left\{
		\begin{aligned}
			x^\prime_i 
				&= \min\ \{ x^\prime_{l} \mid {l \in \cJ^\prime (i)} \} , 
				& 1 \le i \le n_1+n_2, \\
			x^\prime_k 
				&= \begin{bmatrix} \bar \qq _k \\ - {\bar \qq _k} \T \1 + b_k \end{bmatrix} \T \xx^\prime - b_k, 
				& n_1+n_2 < k \le n , \\
			x^\prime_{l} 
				&= -x^\prime_{l-n}, 
				& n < l \le 2n, \\
			x^\prime_{m} 
				&= x^\prime_{m-2n}, 
				& 2n < m \le 3n, \\
			x^\prime _{3n+1} 
				&= 0 ,
		\end{aligned}
		\right.
	\end{equation}
	which satisfies conditions~C\ref{condition:conv_halting},~C\ref{condition:sum-to-1}~and~C\ref{condition:one-type}.
	Then, $\xx$ is a solution to \eqref{eq:Problem} if and only if $\xx^\prime = \begin{bmatrix}
		\tilde \xx \\ - \tilde \xx \\ \tilde \xx \\ 0
	\end{bmatrix} $ is the solution to \eqref{eq:reduction to C 1,3,4}, where
	$$
	\tilde \xx = \begin{bmatrix} \mI _{n_1} & \0 _{n_1 \times (n-n_1)} \\ \0 _{(n-n_1) \times (n-n_1)} & - \mI _{n-n_1} \end{bmatrix} \cdot \xx \,.
	$$
	Therefore, the LEMM decision problem under $\{ \textnormal{ C\ref{condition:conv_halting}, C\ref{condition:sum-to-1}, C\ref{condition:one-type} } \}$ is as hard as the LEMM decision problem under $\{\textnormal{C\ref{condition:conv_halting}}\}$.
\end{proof}

\subsection{Deferred proofs in \cref{sec:checking conditions}}

We first state Lemma~\ref{lemma:nonnegative entries} and give its proof.
\begin{lemma}
	\label{lemma:nonnegative entries}
	Assume $\mQ \in \R _{\ge 0} ^{n \times n} $. The following conditions are equivalent:
	\begin{enumerate}
		\item \label{equiv:halting} $\lim _{m \rightarrow \infty} \mQ^m = \0_{n \times n}$;
		\item \label{equiv:invert} $(\mI-\mQ)$ is invertible and $(\mI-\mQ)^{-1} \ge 0$;
		\item \label{equiv:lp} \( \exists \xx \in \R_{\ge 0}^n \), \st \( \xx \ge \mQ \xx + \1 \).
	\end{enumerate}
\end{lemma}

\begin{proof}[Proof of Lemma~\ref{lemma:nonnegative entries}]
	We give a circular sequence of implications.
	
	(\ref{equiv:halting}$\Rightarrow$\ref{equiv:invert}).
	It is direct from \cref{lemma:linear system existence&uniqueness}.
	
	(\ref{equiv:invert}$\Rightarrow$\ref{equiv:lp}).
	For $\xx = (\mI-\mQ)^{-1} \1 \in \R_{\ge 0}^n$, we have $\xx = \mQ\xx + \1$.
	
	(\ref{equiv:lp}$\Rightarrow$\ref{equiv:invert}).
	Assume \(\xx \in \R_{\ge 0}^n\), \st \(\xx \ge \mQ \xx + \1 \). Then, \(\xx \ge \mQ \xx + \ee_i\), for all $i \in [n]$.
	For all $i \in [n]$, let $\xx^{(i;0)} = \xx$, and $\xx^{(i;t+1)} = \mQ \xx^{(i;t)} + \ee_i$, for all $t = 0, 1, \ldots$
	We have that, for every $t \in \N$,
	$$
	\xx^{(i;0)} \ge \xx^{(i;1)} \ge \ldots \ge \xx^{(i;t)} \ge \ldots \ge 0 \,.
	$$
	Therefore, we can assume the sequence $\left( \xx ^{(i;t)} \right) _{t \in \N}$ converges to $\xx^{(i)}$.
	So, we have
	$$
	\xx^{(i)} \ge 0 \text{ and } \xx^{(i)} = \mQ \xx^{(i)} + \ee_i .
	$$
	Therefore,
	\begin{align*}
		\mI
		&= [\ee_1, \ldots, \ee_n] \\
		&= [(\mI-\mQ)\xx^{(1)}, \ldots, (\mI-\mQ)\xx^{(n)}] \\
		&= (\mI-\mQ) \cdot [\xx^{(1)}, \ldots, \xx^{(n)}] ,
	\end{align*}
	and we have that $(\mI-\mQ)$ is invertible and
	$$
	(\mI-\mQ)^{-1} = [\xx^{(1)}, \ldots, \xx^{(n)}] \ge 0 \,.
	$$
	
	(\ref{equiv:invert}$\Rightarrow$\ref{equiv:halting}).
	Let the spectral radius of $\mQ$ be $\rho$.
	Assume $\rho \ge 1$.
	By the Perron-Frobenius Theorem (non-negative matrix version),
	$$
	\exists \vv \in \R_{\ge 0}^n \setminus \{ \0 \}, \quad \mQ \vv = \rho \vv .
	$$
	In view of $(\mI-\mQ)\vv = \vv - \mQ\vv = (1-\rho) \vv$, we have
	$$
	0 \le \vv = (\mI-\mQ)^{-1} \cdot (1-\rho) \vv \le 0 .
	$$
	So, $\vv = \0$ and we have the contradiction. Therefore, $\rho < 1$, \ie\ $\lim _{m \rightarrow \infty} \mQ^m = \0_{n \times n}$.
\end{proof}

\begin{proof}[Proof of Lemma~\ref{lemma:checking via LP}]
	We present each direction separately.
	
	(if)
	Consider \( \xx \in \R_{\ge 0}^n\), such that, for all \( \mQ \in \cQ\), we have \(\xx \ge \mQ \xx + \1 \).
	Fix any $\mQ = \sum _{i \in \cI} \alpha_i \mQ_i$, in which $\mQ_i \in \cQ,\ \alpha_i \ge 0$, for all $i \in \cI$, and $\sum _{i \in \cI} \alpha_i = 1 $.
	By algebraic manipulations,
	$$
	\mQ \xx + \1 = \sum _{i \in \cI} \alpha_i (\mQ_{i} \xx + \1) \le \sum _{i \in \cI} \alpha_i \xx = \xx \,.
	$$
	Therefore, by \cref{lemma:nonnegative entries}, $\lim _{m \rightarrow \infty} \mQ ^{m} = \0 _{n \times n}$.
	
	(only if)
	By \cref{lemma:existence&uniqueness}, there exists $\xx$ such that
	\begin{equation}
		\label{eq:max construction}
		\xx = \max \ \left\{ \mQ\xx + \1 \mid {\mQ \in \cQ} \right\} .
	\end{equation}
	Therefore, $\xx \ge \mQ \xx + \1$, for all $\mQ \in \cQ$.
	Also,
	$$
	\exists \mQ \in \cQ, \ \st \xx = \mQ\xx + \1,
	$$
	so, by \Cref{lemma:nonnegative entries}, we have
	$$
	\xx = (\mI-\mQ)^{-1} \1 \ge 0 \,.
	$$
\end{proof}

\begin{proof}[Proof of Theorem~\ref{thm:checking c1,3,4 is conp hard}]
	For a \textsc{SAT} instance in conjunctive normal form with variables $v_1, \ldots, v_r$ and clauses $c_1, \ldots, c_m$, define, for all $1 \le i \le m$,
	\begin{align*}
		s(i)
		&= \left\{ m+l \mid v_l \text{ is a literal in } c_i \right\} \cup \\
		&\qquad \left\{ m+r+l \mid \neg v_l \text{ is a literal in } c_i \right\} \,.
	\end{align*}
	We construct the LEMM $(n_1,n_2,n,\cJ,\qq,\bb)$, in which, for all $1 \le i \le m$,  $m < j \le m+r$, and $m+r < k \le m+2r$
	$$
	\begin{alignedat}{2}
		n_1 &= 0, \\
		n_2 &= m+r, \\
		n &= m+2r+2, \\
		\cJ(i) &= s(i), \\
		\cJ(j) &= \{ m+2r+1, m+2r+2 \}, \\
		\qq_{k} &= -\ee_{k-r}, \\
		\qq_{m+2r+1} &= \frac 1 {m+1} \cdot \begin{bmatrix}
			\1_{m} \\
			\0_{2r} \\
			1 \\
			0
		\end{bmatrix} , \\
		\qq_{m+2r+2} &= -\ee_{m+2r+1} , \\
		\bb = \0.
	\end{alignedat}
	$$
	It satisfies conditions~C\ref{condition:sum-to-1}~and~C\ref{condition:one-type}.
	Moreover, we have ``condition~C\ref{condition:conv_halting} holds'' if and only if ``$c_1 \land \ldots \land c_m$ is unsatisfiable''. Because the \textsc{SAT} problem is \textsc{NP}-hard, checking condition~C\ref{condition:conv_halting} is \textsc{coNP}-hard.
\end{proof}

\end{document}